\documentclass[12pt]{llncs}

\usepackage{fullpage}
\usepackage{graphicx}
 \pdfoutput=1
\usepackage{epsf}
\usepackage{epsfig}
\usepackage{epstopdf}
\usepackage{amsfonts}
\usepackage{amssymb}
\usepackage{amsmath}
\usepackage{latexsym}
\usepackage[all]{xy}
\usepackage{hyperref}
\usepackage{nameref}

\makeatletter
\let\orgdescriptionlabel\descriptionlabel
\renewcommand*{\descriptionlabel}[1]{%
  \let\orglabel\label
  \let\label\@gobble
  \phantomsection
  \edef\@currentlabel{#1}%
  \let\label\orglabel
  \orgdescriptionlabel{#1}%
}
\makeatother

\newtheorem{fact}{Fact}
\newtheorem{obs}{Observation}

\newcommand{\Log}{\mbox{{\sf L}}}

\newcommand{\ShP}{\mbox{{\sf \#P}}}

\newcommand{\NP}{\mbox{{\sf NP}}}
\newcommand{\Gl}{\mbox{{\sf GapL}}}

\newcommand{\calC}{\mbox{${\cal C}$}}

\begin{document}
\title{Tree-width and Logspace: \\ Determinants and Counting Euler Tours}
\author{Nikhil Balaji \& Samir Datta}
\institute{
  Chennai Mathematical Institute (CMI), India\\
  \email{\{nikhil,sdatta\}@cmi.ac.in}
}
\maketitle
\begin{abstract}
Motivated by the recent result of \cite{EJT} showing that $\mathsf{MSO}$ properties
are Logspace computable on graphs of bounded tree-width, we consider the complexity
of computing the determinant of the adjacency matrix of a bounded tree-width graph
and prove that it is \Log-complete. It is important to notice that the determinant 
is neither an $\mathsf{MSO}$-property nor counts the number of solutions of an
$\mathsf{MSO}$-predicate.

We extend this technique to count the number of spanning arborescences and 
directed Euler tours in bounded tree-width digraphs, and further to counting 
the number of spanning trees and the number of Euler tours in undirected graphs,
all in \Log. Notice that undirected Euler tours are not known to be 
$\mathsf{MSO}$-expressible and the corresponding counting problem is
 in fact $\ShP$-hard for general graphs. Counting undirected Euler tours in 
bounded tree-width graphs was not known to be polynomial time computable
till very recently Chebolu et al \cite{CCMtw} gave a polynomial time algorithm
for this problem (concurrently and independent of this work).

Finally, we also show some linear algebraic extensions of the determinant 
algorithm to show how to compute the characteristic polynomial and trace of
the powers of a bounded tree-width graph in \Log.
\end{abstract}
\section{Introduction}
The determinant is a fundamental algebraic invariant of a matrix.
For a $n \times n$ matrix $A$ the determinant is given by the expression
\[
 \mbox{Det}(A) = \sum_{\sigma \in S_n} \mbox{sign}(\sigma) \prod_{i \in [n]} a_{i,\sigma(i)}
\]
where $S_n$ is the symmetric group on $n$ elements, $\sigma$ is a permutation
from $S_n$ and $\mbox{sign}(\sigma)$ is the number of inversions in $\sigma$. 
Even though the summation in the definition runs over $n!$ many terms, there are many
efficient sequential \cite{vzGG} and parallel \cite{Berk} algorithms for computing the determinant.

Apart from the inherently algebraic methods to compute the determinant there
is also a combinatorial algorithm (see, for instance, Mahajan and Vinay \cite{MV}) 
which interprets the determinant as a signed sum of cycle covers in the weighted
adjacency matrix of a graph. \cite{MV} are thus able to give another proof
of the \Gl-completeness of the determinant, a result first proved by 
Toda \cite{Toda91}.
For a more complete discussion on the known algorithms for the determinant, see \cite{MV}.


Armed with this combinatorial interpretation
 (and the fact that the determinant is $\Gl$-complete),
one can ask if the determinant is any easier if the underlying matrix represents simpler
classes of graphs. Datta, Kulkarni, Limaye, Mahajan \cite{DKLM} study the complexity of the determinant 
and permanent
respectively, when the underlying graph is planar and show that they are as 
hard as the general case, i.e., respectively $\Gl$ and $\ShP$-hard. 
We revisit these questions in the context of bounded tree-width graphs.
%

Many $\NP$-complete graph problems become tractable when restricted to 
graphs of bounded treewidth. In an influential paper, Courcelle \cite{Cou} proved
that any property of graphs expressible in Monadic Second Order logic 
can be decided in linear time on bounded treewidth
graphs. For example, Hamiltonicity is an $\mathsf{MSO}$ property and hence deciding if a 
bounded treewidth graph has a Hamiltonian cycle can be done in linear time.
More recently Elberfeld, Jakoby, Tantau \cite{EJT} showed that in fact, $\mathsf{MSO}$ 
properties  on bounded treewidth graphs can be decided in \Log. 

We study the Determinant problem when the underlying directed graph has 
bounded treewidth and show that the determinant computation for this case is in
logspace. As a corollary we also compute other linear algebraic invariants of matrix,
viz. the characteristic polynomial, rank and trace of powers of the matrix in logspace. 
We also give a tight bound on the complexity of the problem by showing that it is $\Log$-
hard via a reduction from directed reachability is paths.

\subsection{Spanning Trees}

The parallel complexity of the problem for planar graphs was investigated by
Braverman, Kulkarni and Roy \cite{BKR}, who give tight bounds on the complexity of the problem,
both in general and in the modular setting. They show that the problem is $\Log$-complete 
when the modulus is $2^k$, for constant $k$ and  for any other modulus and in the non-modular
case, the problem is shown to be as hard in the planar case as for the case of arbitrary graphs.
We consider the bounded treewidth case and show that the problem is in $\Log$ for exactly counting
the number of spanning trees in contrast to \cite{BKR}.

\subsection{Euler Tours}
An Euler tour of a graph is a walk on the graph that traverses every
edge in the graph exactly once. Given a graph, deciding if there is an
Euler tour of the graph is quite simple. Indeed, the famous Konisberg
bridges problem that founded graph theory is just a question
of existence of Euler tours on these bridges. Euler settled in the negative
and in the process gave a necessary and sufficient condition for a graph to
be \textit{Eulerian}(A graph is Eulerian if and only if all the vertices 
are of even degree). This gives a simple algorithm to check if a graph is
Eulerian. 

An equally natural question is to ask for the number of distinct Euler 
tours in a graph. For the case of directed graphs, the {\bf BEST} theorem due
to De Bruijn, Ehrenfest, Smith and Tutte gives an exact formula that gives
the number of Euler tours in a directed graph (see Fact \ref{fact:best})
which yields a polynomial time algorithm via a determinant computation.
For undirected graphs, no such closed form expression is known and the
computational problem is \ShP-complete\cite{BW}. In fact, the problem
is \ShP-complete even when restricted to $4$-regular planar graphs
\cite{GS}. So exactly computing the number of Euler tours is not in
polynomial time unless $\ShP = \mathsf{P}$.

This is an instance of an interesting phenomenon in the complexity of 
counting problems : the decision version of the problem is tractable and
the counting version is intractable. The flagship example of this phenomenon
is the Perfect Matching problem for which there is a polynomial time algorithm
\cite{Edm}, whereas the counting version is \ShP-complete via a reduction from
the permanent. Faced with this adversity, traditionally there are two directions
that have been pursued in the community:

\begin{enumerate}
 \item Can one find an good \textit{approximate} solution in polynomial time?
  \footnote{For the Perfect Matching problem, a long line of work culminating in the 
  beautiful algorithm due to Jerrum and Sinclair \cite{JS89a,JS89b} gives a Fully
  Polynomial Randomized Approximation Scheme(FPRAS). But these techniques have not 
  yielded an FPRAS for counting Euler tours.} This problem is wide open for counting
  Euler Tours. 
  \item Find restricted classes of graphs for which one can count exactly
 the number of Euler tours in polynomial time. 
\end{enumerate}

Previously Chebolu, Cryan, Martin  have given a polynomial time algorithm for 
counting Euler tours in undirected series-parallel graphs\cite{CCMsp}.

\subsection{Our Techniqes and Results}

We show that the following can be computed in \Log:
\begin{enumerate}
 \item The Determinant of an $(n \times n)$ matrix $A$ whose underlying 
 undirected graph has bounded treewidth. As a corollary we can also compute
 the coefficients of the characteristic polynomial of a matrix.
 \item Given an $(n \times n)$ matrix $A$ whose underlying 
 undirected graph has bounded treewidth, compute the trace of $A^k$.
 \item Counting the number of Spanning Trees in graphs of bounded treewidth.
 \item Counting the number of Euler tours in a directed graph where
 the underlying undirected graph is bounded treewidth.
 \item Counting the number of Euler tours in a undirected bounded 
 treewidth graph.
\end{enumerate}

At the core of our result is our algorithm to compute the determinant
by writing down an $\mathsf{MSO}_2$ formula that evaluates to true on
every valid cycle cover of the bounded treewidth graph underlying $A$. 
The crucial point being that the cycle covers are parameterised on a 
quantity closely related to the sign of the cycle covers. This makes it possible
to invoke the cardinality version of Courcelle's theorem(for logspace) due 
to \cite{EJT} to compute the determinant. We use this determinant algorithm  
and the BEST theorem to count directed Euler tours. Using the algorithm for
directed Eulerian tours, we get an algorithm to count the number of Euler 
tours in an undirected bounded treewidth graph. Concurrently and independent of
this work, using different techniques, Chebolu, Cryan, Martin \cite{CCMtw} have 
given a polynomial 
time algorithm for counting Euler tours in undirected bounded treewidth graphs. 

\subsection{Organization of the paper}

Section \ref{sec:prelim} introduces some notation and terminology required for the rest of the paper.
In Section \ref{sec:btwDet}, we give a Logspace algorithm to compute the Determinant of 
matrices of bounded treewidth. In Section \ref{sec:appLA}, we give some applications of the determinant
algorithm to computing certain linear algebraic invariants of a matrix. In Section \ref{sec:appG}, we
give some applications of the determinant algorithm to some counting problems on bounded treewidth
graphs. In Section \ref{sec:concl}, we summarize our contributions and mention some interesting open
questions raised by this paper.

\section{Preliminaries}\label{sec:prelim}

\begin{definition}\label{def:tdecomp}
Given an undirected graph $G = (V_G, E_G)$ a tree decomposition of $G$ is a tree
$T = (V_T, E_T)$(the vertices in $V_T \subseteq 2^{V_G}$ are called \textit{bags}), 
such that 
\begin{enumerate}
 \item Every vertex $v \in V_G$ is present in at least one bag, i.e., 
 $\cup_{X \in V_T} X = V_G$.
 \item If $v \in V_G$ is present in bags $X_i, X_j \in V_T$, then
 $v$ is present in every bag $X_k$ in the unique path between $X_i$
 and $X_j$ in the tree $T$.
 \item For every edge $(u, v) \in E_G$, there is a bag $X_r \in V_T$ such that
 $u, v \in X_r$.
\end{enumerate}
The width of a tree decomposition is the $\max_{X \in V_T} (|X| - 1)$. The tree width of
a graph is the minimum width over all possible tree decomposition of the graph.
\end{definition}

\begin{definition}\label{def:ccover}
Given a weighted directed graph $G = (V, E)$ by its adjacency matrix $[a_{ij}]_{i,j\in[n]}$,
a \textit{cycle cover} $\calC$ of $G$ is a set of cycles that cover $G$. 
I.e., $\calC = \{C_1, C_2, \ldots, C_k\}$, where $C_i = \{c_{i_1}, \ldots, c_{i_r}\} \subseteq V$ such that
$(c_{i_1}, c_{i_2}), (c_{i_2}, c_{i_3}), \ldots, (c_{i_{r-1}}, c_{i_r}), (c_{i_r}, c_{i_1}) \in E$ and 
$\cup_{i=1}^k C_i = V$. The weight of the cycle $C_i = \prod_{j \in [r]} \mbox{wt}(a_{ij})$
and the weight of the cycle cover $\mbox{wt}(\calC) = \prod_{i \in [k]} \mbox{wt}(C_i)$. The sign of
the cycle cover $\calC$ is $(-1)^{n+k}$.
\end{definition}


Every permutation $\sigma \in S_n$ can be written as a union of disjoint cycles. Hence a permutation  
corresponds to a cycle cover of a graph on $n$ vertices. In this light, the determinant of an $(n \times n)$
matrix $A$ can be seen as a signed sum of cycle covers:
\[
 \mbox{Det}(A) = \sum_{\mbox{cycle cover\ } \calC} \mbox{sign}(\calC) \mbox{wt}(\calC)
\]


For the definitions of Monadic Second Order logic ($\mathsf{MSO}$) and related terminology we refer the reader to 
Section $2$ of \cite{EJT}.

%

\begin{theorem}{\cite{EJT}}
\label{prop:EJTB}
For every $k \ge 1$, there is a logspace DTM that on input of any graph $G$ of tree width at
most $k$ outputs a width-$k$ tree decomposition of $G$.
\end{theorem}

\begin{theorem}{\cite{EJT}}
\label{prop:EJTC}
For every $k \ge 1$ and every $\mathsf{MSO}$-formula $\phi$, there is a logspace DTM that on input of
any logical structure $\mathcal{A}$ of tree width at most $k$ decides whether $A \vDash \phi$ holds.
\end{theorem}

\begin{theorem}{\cite{EJT}}
Let $k \geq 1$ and let $\phi(X_1, \ldots , X_d)$ be an $\mathsf{MSO}-\tau$-formula. 
Then there are an $s \geq 1$, an $\mathsf{MSO}-\tau_{s-tree}$-formula $\psi(X_1, \ldots, X_d )$
and a logspace DTM that on input of any $\tau$-structure $\mathcal{A}$ with universe $A$
and $tw(\mathcal{A}) \leq k$ outputs a balanced binary $s$-tree structure $\mathcal{T}$ such that for all 
indices $i \in \{0, \ldots , |A|\}^d$ we have histogram$(\mathcal{A}, \phi)[i] = \mbox{histogram}
(\mathcal{T} , \psi)[i]$.
\end{theorem}

\section{Determinants in Bounded Treewidth Matrices} \label{sec:btwDet}
Given a square $\{0,1\}$-matrix $A$, we can view it as the bipartite adjacency
matrix of a bipartite graph $G(A)$. The permanent of this matrix $A$ counts the 
number of perfect matchings in $G(A)$, while the determinant counts the signed
sum of perfect matchings in $G(A)$.

If $G$ is a bounded treewidth graph then we can count the number of perfect
matchings in $G$ in \Log \cite{EJT} (see also \cite{DDN}). Hence the 
complexity of the permament of $A$, above is well understood in this case 
while the complexity of computing the determinant is not clear. 

On the other hand the determinant of a $\{0,1\}$-matrix reduces (say by
a reduction $g_{MV}$ to counting
the number of paths in another graph (see e.g. \cite{MV}). Also counting
$s,t$-paths in a bounded treewidth graph is again in \Log via \cite{EJT}
(see also \cite{DDN}). But the problem with this approach is that that the
graph $g_{MV}(G)$ obtained by reducing a bounded treewidth $G$ is not
bounded treewidth.

However, we can also view view $A$ as the adjacency matrix of a directed
graph $H(A)$. If $H(A)$ has bounded treewidth (which implies that $G(A)$ also 
has bounded treewidth) then we have a way of computing the determinant of $A$. 
To see this, consider the following lemma:

\begin{lemma} \label{lem:thetaX}
There is an $\mathsf{MSO}_2$-formula $\theta(X)$ with a
free variable $X$ that takes values from the set of subsets of vertices
such that $\theta(S)$ is true exactly when $S$ is the set of heads of a cycle
cover of the given graph.
\end{lemma}
Note that the heads of a cycle cover are the least numbered vertices from each
of the cycles of the cycle cover.

Let $T$ be a tree decomposition of $G$ and let 
$C: V_G \rightarrow \{1,\ldots,2k+2\}$ be a coloring of the 
vertices of $G$ such that if $x \neq y$ are two vertices of $G$ in the same bag
of $T$ then $C(x) \neq C(y)$. \cite{EJT} proves that such a coloring exists
and is computable in logspace (in fact, they prove this for 
the slightly stronger
property that if $C(x) = C(y)$, the two vertices $x,y$ cannot even be in
neighboring bags).

We define the relation $ION(x,y)$ which is essentially the ``in-order-next''
relation for the tree decomposition where ties are broken according to the
coloring above. More precisely, we define $ION(x,y)$ to be true if precisely
one of the following holds:
\begin{enumerate}
\item $x,y$ are in the same bag for some bag and $C(x) < C(y)$
\item $x,y$ do not occur in the same bag but $x \in B_x, y \in B_y$ where:
\begin{enumerate}
\item Bag $B_x$ is the left child of bag $B_y$ in the tree decomposition $T$
\item Bag $B_y$ is the right child of bag $B_x$ in the tree decomposition $T$
\end{enumerate}
\end{enumerate}

The following is well known and follows from Euler Traversal of a tree 
in logspace:
(see \cite{CookMckenzie}).
\begin{proposition}
$ION$ is computable in logspace.
\end{proposition}
From the definition of the coloring defined above the following proposition
about the transitive closure $ION^{*}$ of $ION$ is also clear:
\begin{proposition}
$ION^{*}$ is a total order on the vertices of $G$.
\end{proposition}

Finally, the local nature of $ION$ yields that: 
\begin{proposition}
If we add $ION$ to the structure
representing the tree decomposition $T$, its tree width goes up only by a
constant factor. 
\end{proposition}
Since we just need to replace each bag by the union of it with its neighbors
to obtain a tree decomposition such that two vertices related by $ION$ are in 
the same bag.

\begin{proof}{(of Lemma \ref{lem:thetaX})}
We guess a set of edges $Y$ and vertices (heads of a cycle cover) $H$ and 
verify that the subgraph induced by $Y$ indeed forms a cycle cover of $G$.
Our $\mathsf{MSO}_2$ formula $\theta(X)$ is of the form:
\[
(\exists Y \subseteq E) \phi(X,Y) 
\]
where,
\[
\phi(X,Y) \equiv  (\forall v \in V(Y)) (\exists h \in H)\ \left(\mathsf{DEG}(v)\right)
 \bigwedge \left(\mathsf{PATH}(h,v) \implies \mathsf{ION}^{*}(h,v))\right) 
\]

Here $V(Y) = \cup_{y \in Y}{y}$ is the set of vertices on which edges in $Y$
are incident.

\begin{enumerate}
 \item $\mathsf{DEG}(v)$ is the predicate that asserts that the degree of $v$ (in the 
 subgraph induced by the edges in $Y$) is exactly $2$.
 \item $\mathsf{PATH}(x, y)$ is the predicate that asserts that there is a
  path from $x$ to $y$.
\end{enumerate}
One can check that all the predicates above are $\mathsf{MSO}_2$ definable. 
%
\qed
\end{proof}

Now we are done with the help of the following fact (see e.g. \cite{MV}):
\begin{fact}\label{fact:ccsign}
The sign of a cycle cover consisting of $k$ cycles is $(-1)^{n+k}$ where
$n$ is the number of vertices in the graph.
\end{fact}

Thus, using the histogram version of Courcelle's theorem from \cite{EJT} we 
get that:

\begin{theorem}\label{thm:btwDet}
The determinant of a matrix $A$, which can be viewed as the adjacency matrix
of a directed graph of bounded treewidth, is in \Log.
\end{theorem}

\begin{proof}
The histogram version of Courcelle's theorem as described in \cite{EJT}
when applied to the formula $\theta(X)$ above yields the number of cycle covers
of $H(A)$ parametrized on $|X|$. But in the notation of the fact above,
$|X| = k$ so we can easily compute the determinant as the alternating sum of
these counts.\qed
\end{proof}

Our final algorithm for computing the determinant of bounded treewidth matrices 
is the following: On input $0-1$ matrix $A$:

\begin{enumerate}
 \item Compute the tree decomposition $T_A$ of the underlying directed graph $G_A$.
 \item Compute an ordering of the vertices of $G_A$ using the labels $(DFS(T_A), B(T_A))$.
 \item Implement the cardinality version of \cite{EJT} via Theorem \ref{thm:btwDet}.
\end{enumerate}

\begin{lemma}\label{lem:hard}
 For all constant $k \ge 1$, computing the determinant of an $(n \times n)$ matrix $A$
 whose underlying undirected graph has treewidth at most $k$ is $\Log$-hard.
\end{lemma}

\begin{proof}
 We reduce the problem $\mathsf{ORD}$ of deciding for a directed path $P$
 and two vertices $s, t \in V(P)$ if there is a path from $s$ to $t$ (known to be
 $\Log$-complete via \cite{Ete97}) to computing the determinant of bounded treewidth
 matrices (Note that $P$ is a path and hence it has treewidth $1$). Since $\mathsf{ORD}$
 is $\Log$-hard, we get that the determinant of bounded treewidth matrices is $\Log$-hard.
 
 Our reduction is as follows: Given  a directed path $P$ with source $a$, sink $b$ 
 and distinguished vertices $s$ and $t$, we construct a new graph $P'$ as follows: 
 Add edges $(a, s'), (s', t), (t, s), (s, a)$ and  $(b, t')$ and remove edges 
 $(s', s), (t, t')$ where $s'$ and $t'$ are vertices in $P$ such that 
 $(s', s), (t, t') \in E(P)$ (See Figure \ref{fig:hardness}).
 
 We claim that there is a directed path between $s$ and $t$ if and only if  
 the determinant of the adjacency matrix of $P'$ is zero. To see this, notice 
 that if there is a directed path from $s$ to $t$ in $P$, then there are exactly
 two cycle covers in $P'$ -- $(a,s') (s,t) (t',b)$ with three cycles and $(a, s', t, s)
 , (t', b)$ with two cycles. Using Fact \ref{fact:ccsign}, the signed sum of these cycle
 covers is $(-1)^{n+3} + (-1)^{n+2} = 0$, which is exactly the determinant of $P'$.
 
 In the case that $P$ has a directed path from $t$ to $s$, then there is exactly one cycle
 namely $(a, t, s', b, t', s)$. We argue as follows: The edges $(t,s), (s,b), (b,t'), (t',s')$ are in the
 cycle cover since they are the only incoming edges to $s, b, t', s'$ respectively. So $(t, s, b, t', s')$
 is a part of any cycle cover of the graph. This forces one to pick the edge $(s', a)$ and hence we have
 exactly one cycle in the cycle cover for $P'$.
 
\begin{figure}
\begin{center}
\includegraphics[width=\textwidth,height=\textheight,keepaspectratio]{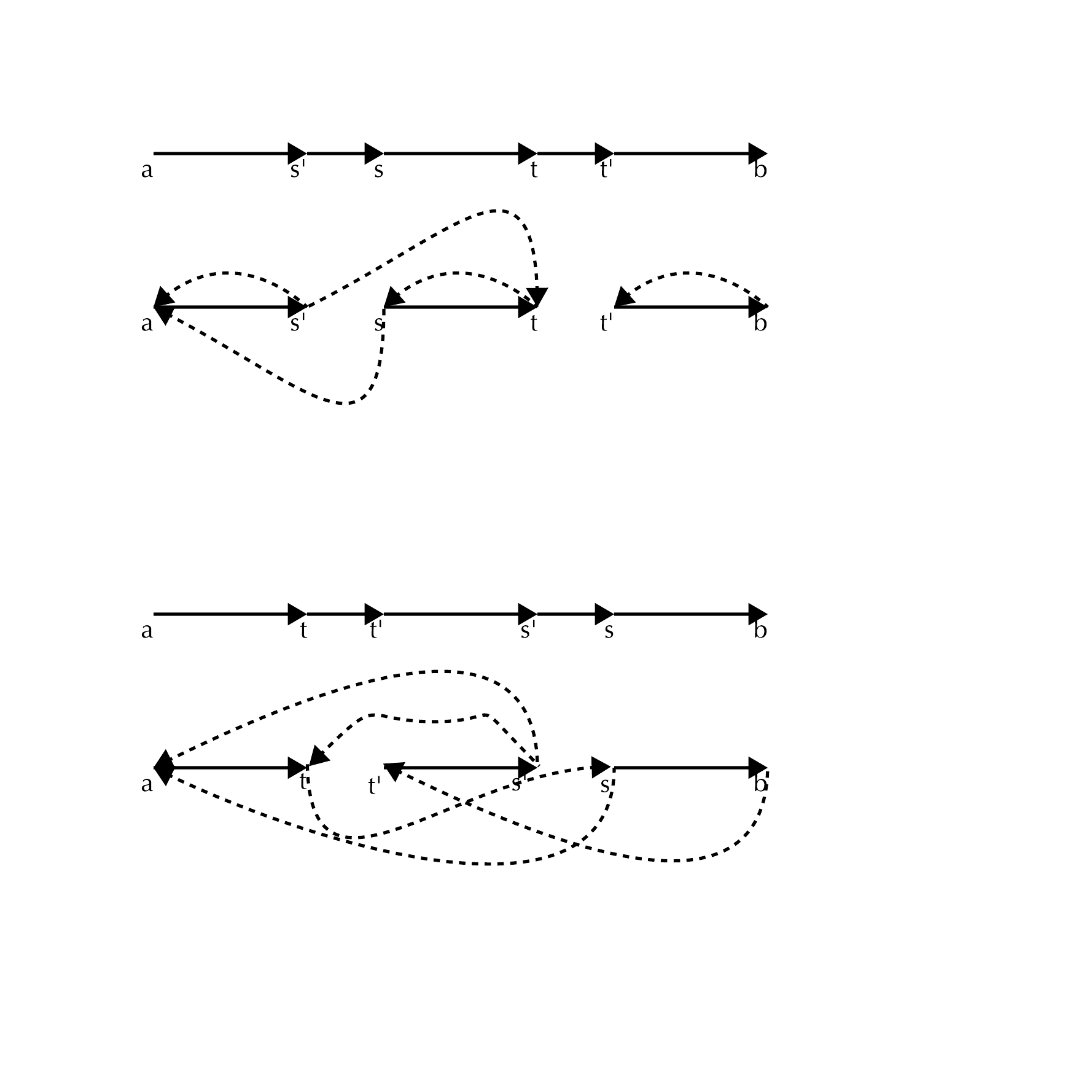}
\caption{$\Log$-hardness for bounded treewidth determinant: Above - $s$ occurs before $t$; Below - $t$ occurs before $s$}
\label{fig:hardness}
\end{center} 
\end{figure}
 
\end{proof}

\section{Applications to Linear Algebraic problems on Bounded Treewidth Matrices}\label{sec:appLA}
%
%

\subsection{Characteristic Polynomial}
\begin{corollary}{(of Theorem \ref{thm:btwDet})}{\label{cor:charpoly}}
There is a logspace machine that takes as input a $(n \times n)$ matrix $A$,
$1^m$ where $1 \leq m \leq n$ and computes the coefficient of $x^m$ in 
characteristic polynomial ($\chi_A(x) = det(xI - A)$) of $A$.
\end{corollary}

\begin{proof}
While counting the number of cycle covers with $k$
cycles, we can keep track of the number of self-loops
occurring in a cycle cover. Hence we can also compute the 
characteristic polynomial in \Log. \qed
\end{proof}

\subsection{Trace of matrix powers}
\begin{theorem}\label{thm:matpow}
 There is a logspace algorithm that on input a $(n \times n)$ matrix $A \in \{0,1\}^{n^2}$
 $1^k$ and $1^i$ computes the $i$-th bit of trace of $A^k$. 
\end{theorem}

We first introduce some notation:
\begin{definition}
Let $S^k_n$ denote the elementary symmetric polynomials (for $n > 0, 0 < k \leq n$) i.e.
\[
S^k_n(X_1,\ldots,X_n) = \sum_{1 \leq j_1 < j_2 < \ldots < j_k \leq n}{X_{j_1}X_{j_2} \ldots X_{j_k}}
\]
By convention, $S^0_n(X_1,\ldots,X_n) =1$ and if $k > n$, the $S^k_n$ is
identically zero.
Let $P^k_n$ denote the $k$-th power sums (for $k > 0, n > 0$) i.e.
\[
P^k_n(X_1,\ldots,X_n) = \sum_{i=1}^n{{X_i}^k}
\]
Note that $P^0_n = \sum_{i=1}^n X_i^0 = n$. 
Also, let $S_n(t)$ be the following univariate polynomial, (where the 
variables $X_i$ are implicit):
\[
S_n(t) = \sum_{i = 0}^n{S^k_n(X_1,\ldots,X_n)t^k}
\]
Similarly, let,
\[
P_n(t) = \sum_{i=1}^{\infty}{(-1)^kP^k_n(X_1,\ldots,X_n)t^k}
\]
\end{definition}

We will often write $S^k_n, P^k_n$ where $X_1,\ldots,X_n$ are understood.

The following is immediate:
\begin{fact}
\begin{enumerate}
\[
S_n(t) = \prod_{i=1}^n{(1 + X_it)} 
\]
\[
P_n(t) = \sum_{i=1}^n{\frac{1}{1 + X_it}} 
\]
\end{enumerate}
\end{fact}
\begin{proof}
Follows straightaway by expanding out the series and
comparing coefficients. \qed 
\end{proof}
These identities are closely related to Newton's identities\cite{macd}.

\begin{proposition}\label{prop:NewtonProp}
\[
tS'_n(t) = S_n(t)\left(n - P_n(t)\right) 
\]
where $S'_n(t)$ denotes the derivative of $S_n(t)$ wrt $t$. 
\end{proposition}
\begin{proof}
Immediate from expanding out the series and
comparing coefficients. \qed 
\end{proof}

\begin{proof}{(of Theorem \ref{thm:matpow})}
We can compute $P^k_n(X_1,\ldots,X_n) = \sum_{i=1}^n{{X_i}^k}$ where
$X_1, \ldots, X_n$ are the eigenvalues of $A^k$ by Proposition \ref{prop:NewtonProp}.
This is achieved by oracle access to the characteristic polynomial of $A$ 
which is computed by a logspace machine as in Corollary \ref{cor:charpoly}. The trace of $A^k$ 
is the number of closed walks of length exactly $k$. That is $\mbox{Trace}(A^k) = \sum_{i=1}^n
w_{ii}^k$ where $w_{ii}^k$ is the number of closed walks of length $k$ that start and end at vertex $i$.
$P_n^k$ itself is computed by polynomial division as a ratio $\frac{nS_n(t)-tS'_n(t)}{S_n(t)}$. 
This can be done in $\Log$ : Substitute $t = 2^{n^2}$, reducing the problem of
polynomial division to integer division which can be done in $\Log$
via \cite{HAB}. Now the $i$-th bit of the coefficient of $t^k$ of 
$P_n(t)$ can be read off the quotient of this integer division via
scaling by a suitable power of $2$, in this case $2^{n^2(n-k) + i}$.\qed
\end{proof}

%


\section{Applications to Counting problems on Bounded Treewidth graphs}\label{sec:appG}

We now document some consequences of Theorem \ref{thm:btwDet} to some 
counting problems on bounded treewidth graphs. We show that counting the 
number of spanning trees and Euler tours of directed graphs and the Euler
tours of undirected graphs is in $\Log$.

\subsection{Counting the number of Spannning Trees in Bounded Treewidth graphs}

Counting the number of spanning trees of an graph(both directed and undirected) was long
known to be in polynomial time via a determinant computation by Kirchoff's\footnote{See Section 1.2 of 
\url{http://math.mit.edu/~levine/18.312/alg-comb-lecture-19.pdf}
for a nice proof of the undirected case via the directed case} matrix tree theorem\cite{RPSalgebraic}:
\begin{fact}\label{fact:arb}
The number of arborescences of a directed graph equals any
cofactor of its Laplacian.
\end{fact}
where the Laplacian of a directed graph $G$ is $D - A$ where $D$
is the diagoinal matrix with the $D_{ii}$ being the out-degree of vertex $i$
and $A$ is the adjacency matrix of the underlying undirected graph.

However, Theorem~\ref{thm:btwDet} holds only for $\{0,1\}$-matrices while
Laplacians contains entries which are from $\{0,-1\}\cup\mathbb{N}$ because
the diagonal entries are the out-degrees of the various vertices and the off-diagonal
non-zero terms are all $-1$'s.

We first consider a generalisation of the determinant of $\{0,1\}$-matrices of
bounded tree-width viz. the determinant of matrices where the entries are from a 
set whose size is a fixed universal constant and the underlying graph consisting
of the non-zero entries of $A$ is of bounded tree-width.
\begin{lemma}\label{lem:detsup}
Let $A$ be a matrix whose entries belong to a set $S$ of
fixed size independent of the input or its length. If the underlying digraph 
with adjacency matrix $A'$, where $A'_{ij} = 1$ iff $A_{ij} \neq 0$, is of 
bounded tree-width then the determinant of $A$ can be computed in \Log.
\end{lemma}
\begin{proof}
Let $s = |S|$ be a universal constant, $S = \{c_1,\ldots,c_s\}$
 and let $\mbox{val}_i$ be the predicates that partitions the edges of 
$G$ according to their values i.e.  $\mbox{val}_i(e)$ is true iff the edge $e$ has 
value $c_i \in S$.
Our modified formula $\psi(X,Y_1,\ldots,Y_s)$ will contain $s$ unquantified 
new edge-set variables $Y_1,\ldots, Y_s$ along with the old vertex variable $X$,
and is given by:
\[
\forall{e \in E}{\left(\left(e \in Y_i\rightarrow \mbox{val}_i(e)\right) \wedge (e \in Y \leftrightarrow \vee_{i=1}^s{(e \in Y_i)}\wedge \phi(X,Y)\right)}
\]
Notice that we verify that the edges in the set $Y_i$ belong to 
the $i^{\mbox{th}}$ partition and each eadge in $Y$ is in one of the $Y_i$'s.
The fact that the $Y_i$'s form a partition of $Y$ follows from the assumption
that $\mbox{val}_i(e)$ is true for exactly one $i \in [s]$ for any edge $e$.

To obtain the determinant we consider the histogram parameterised on the $s$
variables $Y_1,\ldots,Y_s$ and the heads $X$. For an entry indexed by
$x,y_1,\ldots,y_s$, we multiply the entry by 
$(-1)^{n+x} \prod_{i = 1}^s{{c_i}^{y_i}}$ and take a sum over all entries.
\qed
\end{proof}
Now the proof is completed by observing the following:
\begin{lemma}\label{lem:super}
Given a digraph $G$ there is an easily constructible super-digraph $G'$ such that
the number of rooted spanning trees in both the digraphs is the same and the outdegree of every
vertex in $G$ is either $1$ or $|V(G)|$. In addition if $G$ is bounded tree-width so is $G'$ and
the latter has a tree-decomposition easily obtainable from that of $G$.
\end{lemma}
\begin{proof}
The construction of $G'$ is simple for every vertex $v$ of $G$ which has out-degree $od(v)$, add
$k = n - od(v)$ many new vertices $v_1,v_2,\ldots,v_k$ in $G'$ and ensure that the edges $(v_i,v),(v,v_i)$
are added also. Then $od_{G'}(v_i) = 1$ and $od_{G'}(v) = n$ in $G'$. The tree-width of $G'$ is clearly
same as that of $G$ merely by introducing new bags containing $v,v_i$ for each $i \in [k]$.\qed
\end{proof}
The previous fact implies that the Laplacian matrix of $G'$ has entries
from the set $S = \{0,1,-1, |V(G)|, -|V(G)|\}$.

\subsection{Counting Euler Tours in Directed Bounded Treewidth Graphs}
The BEST Theorem states:
\begin{fact}\cite{BE87}\cite{ST41}\label{fact:best}
The number of Euler Tours in a directed Eulerian graph $K$ is exactly:
\[
t(K)\prod_{v \in V}{(\mbox{deg}(v) - 1)!}
\]
where $t(K)$ is the number of arborescences in $K$ rooted at an arbitrary
vertex of $K$ and $\mbox{deg}(v)$ is the indegree as well as the outdegree
of the vertex $v$.
\end{fact}

We combine Facts \ref{fact:arb} and \ref{fact:best}with Theorem~\ref{thm:btwDet}
to compute the number of directed Euler Tours in a directed Eulerian graph.
Combining Lemmas \ref{lem:detsup} and \ref{lem:super}  with the Theorem~\ref{thm:btwDet} and 
the fact that iterated integer sum and product is in \Log \cite{HAB}, we get 
that:
\begin{theorem}\label{thm:btwDirEuler}
Counting the number of directed Euler Tours in a directed Eulerian graph $G$
is in \Log.
\end{theorem}

\subsection{Counting Euler tours in undirected bounded treewidth graphs}
There is no closed form expression for the number of Euler Tours in undirected
graphs unlike the BEST formula for directed graphs. However, notice that
every undirected tour corresponds to an orientation of the undirected
graph in the following sense. Suppose we orient the edges of the tour so that
the first edge is the lex-least edge (i.e. from the smallest numbered vertex to 
its least numbered neighbour), then this induces an orientation of the edges 
of the undirected graph. Also given any directed tour of any orientation of the
undirected graph the corresponding undirected tour imposes an orientation which
is either the given orientation (if the lex-least edge is oriented according
to the orientation) or its reverse (if the lex-leats edge is oriented according
tot he reverse of the orientation).

The long and short of it is that
\begin{obs} \label{obs:undirDir} The number of undirected tours is 
half of the sum of the directed tours over all Eulerian-orientations of the
graph. 
\end{obs}

Let $G = (V,E)$ be an undirected graph and let $G'$ denote the graph obtained by
subdividing every edge by introducing a new vertex $v_e$ for every 
$e \in E$ and replacing the edge $e = \{u,v\}$ by the two edges $\{u,v_e\},
\{v,v_e\}$. Notice the following bijection:
\begin{obs}
The orientations of $G$ are in bijection with functions that pick for each
vertex $v_e$ of $V(G')$ one of the two edges incident on it. In particular, 
Eulerian orientations correspond to functions such that the edges picked at a 
vertex $v \in V$ are exactly equal to the edges not picked at $v$.
\end{obs}
To see this we just associate the direction $(u,v)$ to $\{u,v\}$ iff
the edge $\{u,v_e\}$ is picked in the set. Under this bijection, a directed
cycle in an oriented version of $G$ consists of a cycle in $G'$ in which
alternating edges are picked by the function above.

Also notice that $G'$ has the same number of Euler tours as $G$ and also
has bounded treewidth.

With all this in mind we can write an MSO-formula,
$\phi(X,Z)$ where $Z$ is the set of edges picked above to indicate an 
orientation. The formula is similar to the one in Lemma~\ref{lem:phiX},
replacing directed cycles/paths and alternating ones (in the sense above).
Notice we don't have to stipulate anything special such as per-vertex
cardinality constraints on $Z$ - this is guaranteed by the rest of the
formula.

Now, the following theorem is obvious in light of 
Observation~\ref{obs:undirDir}:
\begin{theorem}\label{thm:btwUndirEuler}
Counting the number of Euler Tours in an undirected bounded tree-width Eulerian graph $K$
is in \Log.
\end{theorem}

\section{Conclusion and Open Ends}\label{sec:concl}
We show that the Logspace version of Courcelle's theorem can be used to
compute the Determinant of a $\{0,1\}$-matrix of bounded tree-width in \Log.
With some more work we are able to show some problems in linear algebra
such as characteristic polynomial, rank and trace of powers of binary matrices
is also in \Log. We are also able to show that counting graph theoretic structures
such as spanning trees and Euler tours in both graphs and digraphs is in \Log.

Many open questions remain:
\begin{itemize}
\item What is the complexity of the determinant of a matrix with arbitrary weights
whose support is a bounded tree-width graph?
\item What is the complexity of powering a matrix of bounded tree-width?
\item More generally, what is the complexity of iterated product of a bounded
tree-width matrix?
\item What is the complexity of other linear algebraic invariants such as
minimal polynomial of a bounded tree-width matrix?
\item Can the algorithms we design be made less ``galactic'' in terms of dependence
on the tree-width? In other words, can we eliminate the use of Courcelle's 
theorem in these algorithms (\cite{CCMtw} is a step in this direction)?
\end{itemize}

\section*{Acknowledgement}
We would like to thank Abhishek Bhrushundi, Johannes K\"obler, Sebastian Kuhnert,
Raghav Kulkarni, Arne Meier and Heribert Vollmer for illuminating discussions
regarding this paper.

\bibliographystyle{alpha}	
\bibliography{skeleton}
\end{document}